\documentclass{article}

\usepackage{makeidx}

\usepackage[ruled,vlined]{algorithm2e}
\usepackage{amsmath,amssymb}

\newcommand{\eps}{\varepsilon}

\newcommand{\remove}[1]{}
\newcommand{\stream}{{\cal S}}
\newcommand{\sample}{{\cal P}}

\newcommand{\cY}{\mathcal{Y}}
\newcommand{\cZ}{\mathcal{Z}}
\newcommand{\cI}{\mathcal{I}}

\newcommand{\sthminimum}{u}
\newcommand{\samplesize}{s}
\newcommand{\nepochs}{\xi}
\newcommand{\mesg}{\mu}

\newtheorem{theorem}{Theorem}

\newtheorem{lemma}{Lemma}

\newtheorem{definition}{Definition}

\newcommand{\qedsymb}{\hfill{\rule{2mm}{2mm}}}

\newenvironment{proof}{\begin{trivlist}
\item[\hspace{\labelsep}{\bf\noindent Proof: }]
}{\qedsymb\end{trivlist}}

\begin{document}

\title{Optimal Random Sampling from Distributed Streams Revisited \footnote{This writeup is a revised version of a paper with the same title and authors, which appeared in the Proceedings of the International Conference on Distributed Computing (DISC) 2011. It corrects an error in the proof of the upper bound on message complexity (Section 4). The proofs in pages 9, 10, and 11 (excluding Theorem 2) have been rewritten relative to the DISC 2011 version. None of the main theorem statements (Theorems 2,3,4) have changed from the DISC 2011 version.}
\footnote{We thank Rajesh Jayaram for pointing out the error in the conference version.}
}

\author{
Srikanta Tirthapura, Iowa State University, {\tt snt@iastate.edu}
\and
David~P.~Woodruff, Carnegie Mellon University, {\tt dwoodruf@cs.cmu.edu}
}

\date{}

\maketitle

\begin{abstract}
We give an improved algorithm for drawing a random sample from a large
data stream when the input elements are distributed across multiple sites
which communicate via a central coordinator. At any point in time the
set of elements held by the coordinator represent a uniform random
sample from the set of all the elements observed so far. When compared
with prior work, our algorithms asymptotically improve the total
number of messages sent in the system as well as the computation
required of the coordinator. We also present a matching lower bound,
showing that our protocol sends the optimal number of messages up to a
constant factor with large probability. As a byproduct, we obtain an
improved algorithm for finding the heavy hitters across multiple
distributed sites.
\end{abstract}

\section{Introduction}

For many data analysis tasks, it is impractical to collect all the
data at a single site and process it in a centralized manner. For
example, data arrives at multiple network routers at extremely high
rates, and queries are often posed on the union of data observed at
all the routers. Since the data set is changing, the query results
could also be changing continuously with time.  This has motivated the
{\it continuous, distributed, streaming model} \cite{cmy08}. In this
model there are $k$ physically distributed sites receiving high-volume
local streams of data. These sites talk to a central coordinator, who
has to continuously respond to queries over the union of all streams
observed so far. The challenge is to minimize the communication
between the different sites and the coordinator, while providing an
accurate answer to queries at the coordinator at all times.

A fundamental problem in this setting is to obtain a random sample
drawn from the union of all distributed streams. This 
generalizes the classic {\em reservoir sampling} problem (see, e.g.,
\cite{k81}, where the algorithm is attributed to Waterman; see also
\cite{Vitt85}) to the setting of multiple distributed streams, and has applications to
approximate query answering, selectivity estimation, and query
planning. For example, in the case of network routers, maintaining a
random sample from the union of the streams is valuable for network
monitoring tasks involving the detection of global properties
\cite{HNGHJJT07}. Other problems on distributed stream processing,
including the estimation of the number of distinct elements
\cite{cg05,cmy08} and heavy hitters \cite{bo03,kcr06,msdo05,yz09}, use
random sampling as a primitive. 

The study of sampling in distributed streams was initiated by Cormode
{\it et al} \cite{cmyz10}. Consider a set of $k$ different streams
observed by the $k$ sites with the total number of current items in
the union of all streams equal to $n$.  The authors in \cite{cmyz10}
show how $k$ sites can maintain a random sample of $s$ items without
replacement from the union of their streams using an expected
$O((k+s)\log n)$ messages between the sites and the central
coordinator.  The memory requirement of the central coordinator is $s$
machine words, and the time requirement is $O((k+s)\log n)$. The
memory requirement of the remote sites is a single machine word with
constant time per stream update. Cormode {\it et al}. also prove that
the expected number of messages sent in any scheme is $\Omega(k +
s \log (n/s))$.  Each message is assumed to be a single machine word,
which can hold an integer of magnitude $(kns)^{O(1)}$.\\

\noindent{\bf Notation.} All logarithms are to the base 2 unless
otherwise specified. Throughout the paper, when we use asymptotic
notation, the variable that is going to infinity is $n$, and $s$ and
$k$ are functions of $n$.

\subsection{Our Results}
Our main contribution is an algorithm for sampling without
replacement from distributed streams, as well as a matching lower bound
showing that the message complexity of our algorithm is optimal. A
summary of our results and a comparison with earlier work is shown in
Figure \ref{table:comparison}.\\

\noindent{\bf New Algorithm:}
We present an algorithm which uses an expected 
$$O \left (\frac{k \log (n/s)}{\log (1+(k/s))} \right )$$ number of
messages for continuously maintaining a random sample of size $s$ from
$k$ distributed data streams of total size $n$. Notice that if $s < k/8$, this number is
$O \left (\frac{k \log(n/s)}{\log(k/s)} \right)$, 
while if $s \geq k/8$, this number is $O(s \log (n/s))$.


The memory requirement in our protocol at the central 
coordinator is $s$ machine words, and the time requirement
is $O \left (\frac{k \log n/s}{\log (1+k/s)} \right )$. 
The former is the same as that in the protocol of 
\cite{cmyz10}, while the latter improves their $O((k+s)\log n)$ time requirement.
The remote sites in our scheme store a single machine word 
and use constant time per stream update, which is clearly optimal.

Our result leads to a significant improvement in the message
complexity in the case when $k$ is large. For example, for the basic
problem of maintaining a single random sample from the union of
distributed streams ($s=1$), our algorithm leads to a factor of
$O(\log k)$ decrease in the number of messages sent in the system over
the algorithm in \cite{cmyz10}.

Our algorithm is simple, and only requires the central coordinator to
communicate with a site if the site initiates the communication. This
is useful in a setting where a site may go offline, since it does not
require the ability of a site to receive broadcast messages.

\begin{figure*}
\begin{center}
\begin{tabular}{|c|c|c|c|c|}
\hline
& \multicolumn{2}{c|}{Upper Bound} & \multicolumn{2}{c|}{Lower Bound}\\
\cline{2-5}
& Our Result & Cormode {\em et al.} & Our Result & Cormode {\em et al.}\\
\hline\hline
$s < \frac{k}{8}$ & $O \left (\frac{k \log(n/s)}{\log(k/s)} \right)$  
                  & $O(k \log n)$ 
                  & $\Omega \left (\frac{k \log(n/s)}{\log(k/s)} \right)$ 
                  & $\Omega (k + s \log n)$ \\
\hline
$s \ge \frac{k}{8}$ & $O(s \log(n/s))$ 
                    & $O(s \log n)$
                    & $\Omega(s \log(n/s))$ 
                    & $\Omega(s \log(n/s))$ \\
\hline
\end{tabular}
\caption{Summary of Our Results for Message Complexity of Sampling Without Replacement}
\label{table:comparison}
\end{center}
\end{figure*}

\noindent{\bf Lower Bound:}
We also show that for any constant $q > 0$, any 
correct protocol must send 
$\Omega \left (\frac{k \log (n/s)}{\log (1+(k/s))} \right )$ messages with probability
at least $1-q$. This also yields a bound of 
$\Omega \left (\frac{k \log (n/s)}{\log (1+(k/s))} \right )$ on the expected message
complexity of any correct protocol, showing the expected number of
messages sent by our algorithm is optimal, upto constant factors.

In addition to being quantitatively stronger than the lower bound of
\cite{cmyz10}, our lower bound is also qualitatively stronger, because
the lower bound in \cite{cmyz10} is on the expected number of messages
transmitted in a correct protocol. However, this does not rule out the
possibility that with large probability, much fewer messages are sent
in the optimal protocol. In contrast, we lower bound the number of
messages that must be transmitted in any protocol $99\%$ of the
time. Since the time complexity of the central coordinator is at least
the number of messages received, the time complexity of our protocol
is also optimal.\\

\noindent{\bf Sampling With Replacement.}
We also show how to modify our protocol to obtain a 
random sample of $s$ items from $k$ distributed streams
with replacement. Here we achieve a protocol with
$O\left (\left(\frac{k}{\log(2+(k/(s\log s)))}  + s \log s \right) \log n \right)$ 
messages, improving the $O((k+s \log s)\log n)$-message protocol of \cite{cmyz10}. 
We obtain the same improvement in the time complexity of the central coordinator. \\

\noindent{\bf Heavy-Hitters.}
As a corollary, we obtain a protocol for estimating the heavy hitters 
in distributed streams with the best known message complexity. In this problem
we would like to find a set $H$ of items so that if an element $e$ occurs at least an $\eps$
fraction of times in the union of the streams, then $e \in H$, and if $e$ occurs less than
an $\eps/2$ fraction of times in union of the streams, then $e \notin H$. 
It is known that $O(\eps^{-2} \log n)$
random samples suffice to estimate the set of heavy hitters with high probability, 
and the previous best algorithm \cite{cmyz10} was obtained
by plugging $s = O(\eps^{-2} \log n)$ into a protocol for distributed sampling. We thus improve the 
message complexity from  $O((k+\eps^{-2} \log n)\log n)$ to 
$O\left( \frac{k \log (\eps n)}{\log (\eps k)} + \eps^{-2}\log (\eps n) \log n \right)$. 
This can be significant when $k$ is large compared to $1/\eps$. 

\subsection{Related Work}
In addition to work discussed above, other research in the
continuous distributed streaming model includes estimating frequency moments
and counting the number of distinct elements \cite{cg05,cmy08}, and
estimating the entropy \cite{abc09}.
The reservoir sampling technique has been used extensively in large
scale data mining applications, see for example \cite{DN06,MV07,A06}.
Stream sampling under sliding windows has been considered in
\cite{BOZ09,BDM02}. Deterministic algorithms for 
heavy-hitters over distributed streams, and corresponding 
lower bounds were considered in \cite{yz09}.

Stream sampling under sliding windows over distributed streams has
been considered in \cite{cmyz10}. Their algorithm for sliding windows
is already optimal upto lower-order additive terms (see Theorems 4.1
and 4.2 in \cite{cmyz10}). Hence our improved results for the
non-sliding window case do not translate into an improvement for the
case of sliding windows.

A related model of distributed streams was considered in
\cite{GT01,GT02}. In this model, the coordinator was not required to
continuously maintain an estimate of the required aggregate, but when
the query was posed to the coordinator, the sites would be contacted
and the query result would be constructed. In their model, the
coordinator could be said to be ``reactive'', whereas in the model
considered in this paper, the coordinator is ``pro-active''.\\


\noindent{\bf Roadmap:} We first present the model and problem definition in
Section \ref{sec:model}, and then the algorithm followed by a proof of
correctness in Section \ref{sec:algo}. The analysis of message
complexity and the lower bound are presented in Sections \ref{sec:ub}
and \ref{sec:lb} respectively, followed by an algorithm for sampling
with replacement in Section \ref{sec:withrep}.

\section{Model}
\label{sec:model}
Consider a system with $k$ different sites, numbered from $1$ till
$k$, each receiving a local stream of elements. Let $\stream_i$ denote
the stream observed at site $i$.  There is one ``coordinator'' node,
which is different from any of the sites. The coordinator does not
observe a local stream, but all queries for a random sample arrive at
the coordinator. Let $\stream = \cup_{i=1}^n \stream_i$ be the entire
stream observed by the system, and let $n = |\stream|$.
The sample size $\samplesize$ is a parameter supplied to the
coordinator and to the sites during initialization.  

The task of the coordinator is to continuously maintain a random
sample $\sample$ of size $\min\{n, \samplesize\}$ consisting of
elements chosen uniformly at random without replacement from
$\stream$. The cost of the protocol is the number of messages
transmitted.

We assume a synchronous communication model, where the system
progresses in ``rounds''. In each round, each site can observe one
element (or none), and send a message to the coordinator, and receive
a response from the coordinator.  The coordinator may receive up to
$k$ messages in a round, and respond to each of them in the same
round. This model is essentially identical to the model assumed in
previous work \cite{cmyz10}.  Later we discuss how to handle the case
of a site observing multiple elements per round.

The sizes of the different local streams at the sites, their order of
arrival, and the interleaving of the streams at different sites, can
all be arbitrary. The algorithm cannot make any assumption about
these.

\section{Algorithm}
\label{sec:algo}
The idea in the algorithm is as follows. Each site associates a random
``weight'' with each element that it receives. The coordinator then
maintains the set $\sample$ of $s$ elements with the minimum weights
in the union of the streams at all times, and this is a random sample
of $\stream$. This idea is similar to the spirit in all centralized
reservoir sampling algorithms. In a distributed setting, the
interesting aspect is at what times do the sites communicate with the
coordinator, and vice versa.

In our algorithm, the coordinator maintains $\sthminimum$, which is
the $s$-th smallest weight so far in the system, as well as the sample
$\sample$, consisting of all the elements that have weight no more
than $\sthminimum$. Each site need only maintain a single value
$\sthminimum_i$, which is the site's view of the $s$-th smallest
weight in the system so far. Note that it is too expensive to keep the
view of each site synchronized with the coordinator's view at all
times -- to see this, note that the value of the $s$-th smallest
weight changes $O(s \log (n/s))$ times, and updating every site each time
the $s$-th minimum changes takes a total of $O(sk \log (n/s))$ messages.

In our algorithm, when site $i$ sees an element with a weight smaller
than $\sthminimum_i$, it sends it to the central coordinator. The
coordinator updates $\sthminimum$ and $\sample$, if needed, and then
replies back to $i$ with the current value of $\sthminimum$, which is
the true minimum weight in the union of all streams. Thus each time a
site communicates with the coordinator, it either makes a change to
the random sample, or, at least, gets to refresh its view of
$\sthminimum$.

The algorithm at each site is described in Algorithms
\ref{algo:site-init} and \ref{algo:site-process}. 
The algorithm at the coordinator is described in Algorithm \ref{algo:coord}.

\begin{algorithm}[H]
\caption{Initialization at Site $i$.}
\label{algo:site-init}
\tcc{$\sthminimum_i$ is site $i$'s view of the $s$-th smallest weight
in the union of all streams so far. Note this may ``lag'' the value
stored at the coordinator.}
$\sthminimum_i \gets 1$\;
\end{algorithm}

\begin{algorithm}[H]
\caption{When Site $i$ receives element $e$.}
\label{algo:site-process}

Let $w(e)$ be a randomly chosen weight between $0$ and $1$\;

\If{$w(e) < \sthminimum_i$}
{
   Send $(e, w(e))$ to the Coordinator and receive $\sthminimum'$ from
   Coordinator\;

   Set $\sthminimum_i \gets \sthminimum'$\;
}
\end{algorithm}

\begin{algorithm}[H]
\caption{Algorithm at Coordinator.}
\label{algo:coord}

\tcc{The random sample $\sample$ consists of tuples $(e,w)$ where $e$ is an element, and
$w$ the weight, such that the weights are the $s$ smallest
among all the weights so far in the stream}
$\sample \gets \phi$\;

\tcc{$\sthminimum$ is the value of the $s$-th smallest weight in the
stream observed so far. If there are less than $s$ elements so far, 
then $\sthminimum$ is $1$.}
$\sthminimum \gets 1$\;

\While{true}
{
  \If{a message $(e_i, \sthminimum_i)$ arrives from site $i$}
  {
    \If{$\sthminimum_i < \sthminimum$}
    {
      Insert $(e_i, \sthminimum_i)$ into $\sample$\;

      \If {$|\sample| > \samplesize$}
      {
        Discard the element $(e,w)$ from $\sample$ with the largest weight\;

        Update $\sthminimum$ to the current largest weight in $\sample$ (which
        is also the $\samplesize$-th smallest weight in the entire stream)\;
      }
    }

    Send $\sthminimum$ to site $i$\;
  }

  \If{a query for a random sample arrives}{return $\sample$}
}
\end{algorithm}

\subsection{Correctness}

The following two lemmas establish the correctness of the algorithm.

\begin{lemma}
\label{lem:s-smallest}
Let $n$ be the number of elements in $\stream$ so far. (1) If $n \le s$, then the set $\sample$ 
at the coordinator contains all the $(e,w)$ pairs seen at all the sites so far.
(2) If $n > s$, then $\sample$ at the coordinator consists of 
the $\samplesize$ $(e,w)$ pairs such that the weights of the pairs in $\sample$
are the smallest weights in the stream so far.
\end{lemma}

\begin{proof}
The variable $\sthminimum$ is stored at the coordinator, and
$\sthminimum_i$ is stored at site $i$.  First we note that the
variables $\sthminimum$ and $\sthminimum_i$ are non-increasing with
time; this can be verified from the algorithms.

Next, we note that for every $i$ from $1$ till $k$, at every round,
$\sthminimum_i \ge \sthminimum$.  This can be seen because initially,
$\sthminimum_i = \sthminimum = 1$, and $\sthminimum_i$ changes only in
response to receiving $\sthminimum$ from the coordinator.

Thus, if fewer than $s$ elements have appeared in the stream so far,
$\sthminimum$ is $1$, and hence $\sthminimum_i$ is also $1$ for each
site $i$.  The next element observed in the system is also sent to the
coordinator.  Thus, if $n \le s$, then the set $\sample$ consists of
all elements seen so far in the system.

Next, we consider $n > s$. Note that $\sthminimum$
maintains the $s$-th smallest weight seen at the coordinator, and $\sample$
consists of the $s$ elements seen at the coordinator with the smallest weights.
We only have to show that if an element $e$, observed at site $i$
is such that $w(e) < \sthminimum$
then $i$ must have sent $(e,w(e))$ to the coordinator. This follows because 
$\sthminimum_i \ge \sthminimum$ at all times, and if $w(e) < \sthminimum$, then
it must be true that $w(e) < \sthminimum_i$, and in this case, $(e,w(e))$ is
sent to the coordinator.
 
\end{proof}

\begin{lemma}
\label{lem:correctness}
At the end of each round, sample $\sample$ at the coordinator consists
of a uniform random sample of size $\min\{n,\samplesize\}$ chosen
without replacement from $\stream$.
\end{lemma}

\begin{proof}
In case $n \le \samplesize$, then from Lemma \ref{lem:s-smallest}, we know that 
$\sample$ contains every element of $\stream$. In case $n > \samplesize$, from 
Lemma \ref{lem:s-smallest}, it follows that $\sample$ consists of $\samplesize$ elements
with the smallest weights from $\stream$. Since the weights are assigned randomly,
each element in $\stream$ has a probability of $\frac{\samplesize}{n}$ of belonging in
$\sample$, showing that this is an uniform random sample. Since an
element can appear no more than once in the sample, this is a sample
chosen without replacement.
 
\end{proof}

\section{Analysis of the Algorithm (Upper Bound)}
\label{sec:ub}
We now analyze the message complexity of the maintenance of a random
sample. 

For the sake of analysis, we divide the execution of the
algorithm into ``epochs'', where each epoch consists of a sequence of
rounds. The epochs are defined inductively. Let $r > 1$ be a
parameter, which will be fixed later.  Recall that $\sthminimum$ 
is the $s$-th smallest weight so far in the system (if there are fewer
than $s$ elements so far, $\sthminimum=1$).
Epoch $0$ is the set of all
rounds from the beginning of execution until (and including) the
earliest round where $\sthminimum$ is $\frac{1}{r}$ or smaller. Let
$m_i$ denote the value of $\sthminimum$ at the end of epoch
$i-1$. Then epoch $i$ consists of all rounds subsequent to epoch $i-1$
until (and including) the earliest round when $\sthminimum$ is $\frac{m_i}{r}$
or smaller. Note that the algorithm does not need to be aware of the
epochs, and this is only used for the analysis.

Suppose we call the original distributed algorithm described in Algorithms \ref{algo:coord} and
\ref{algo:site-process} as Algorithm $A$.
For the analysis, we consider a slightly different distributed
algorithm, Algorithm $B$, described below. {\em Algorithm $B$ is
identical to Algorithm $A$ except for the fact that at the beginning
of each epoch, the value $\sthminimum$ is broadcast by the coordinator
to all sites.}

While Algorithm $A$ is natural, Algorithm $B$ is easier to analyze.
We first note that on the same inputs, the value of $\sthminimum$ (and
$\sample$) at the coordinator at any round in Algorithm $B$ is
identical to the value of $\sthminimum$ (and $\sample$) at the
coordinator in Algorithm $A$ at the same round. Hence, the
partitioning of rounds into epochs is the same for both algorithms,
for a given input. The correctness of Algorithm $B$ follows from the
correctness of Algorithm $A$.  The only difference between them is in
the total number of messages sent. In $B$ we have the property that
for all $i$ from $1$ to $k$, $\sthminimum_i = \sthminimum$ at the
beginning of each epoch (though this is not necessarily true
throughout the epoch), and for this, $B$ has to pay a cost of at least
$k$ messages in each epoch.

\begin{lemma}
The number of messages sent by Algorithm $A$ for a set of input
streams $\stream_j, j=1\ldots k$ is never more than twice the number
of messages sent by Algorithm $B$ for the same input.
\end{lemma}

\begin{proof}
Consider site $v$ in a particular epoch $i$. 
In Algorithm $B$, $v$ receives $m_i$ at the beginning of
the epoch through a message from the coordinator. 
In Algorithm $A$, $v$ may not know $m_i$ at the beginning of epoch $i$. 
We consider two cases.

Case I: $v$ sends a message to the coordinator in epoch $i$ in
Algorithm $A$.  In this case, the first time $v$ sends a message to
the coordinator in this epoch, $v$ will receive the current value of
$\sthminimum$, which is smaller than or equal to $m_i$. This communication costs two
messages, one in each direction.  Henceforth, in this epoch, the
number of messages sent in Algorithm $A$ is no more than those sent in
$B$. In this epoch, the number of messages transmitted to/from $v$ in
$A$ is at most twice the number of messages as in $B$, which has at
least one transmission from the coordinator to site $v$.

Case II: $v$ did not send a message to the coordinator in this epoch,
in Algorithm $A$. In this case, the number of messages sent in this
epoch to/from site $v$ in Algorithm $A$ is smaller than in Algorithm
$B$.
 
\end{proof}

Let $\nepochs$ denote the total number of epochs.

\begin{lemma}
\label{lem:num-epochs}
If $r \ge 2$, 
$$
E[\nepochs] \le \left(\frac{\log(n/\samplesize)}{\log r}\right) + 2
$$
\end{lemma}

\begin{proof}
Let $z = \left(\frac{\log(n/\samplesize)}{\log r}\right)$. First, we
note that in each epoch, $\sthminimum$ decreases by a factor of at
least $r$. Thus after $(z+\ell)$ epochs, $\sthminimum$ is no more than
$\frac{1}{r^{z+\ell}} = (\frac{s}{n}) \frac{1}{r^\ell}$.
Thus, we have

$$
\Pr[\nepochs \ge z+\ell] 
  \le \Pr\left[\sthminimum \le \left(\frac{s}{n}\right) \frac{1}{r^\ell}\right]
$$

Let $Y$ denote the number of elements (out of $n$) that have been
assigned a weight of $\frac{s}{n r^{\ell}}$ or lesser. $Y$ is a
binomial random variable with expectation $\frac{s}{r^{\ell}}$.
Note that if $\sthminimum \le \frac{s}{n r^{\ell}}$, it must be true
that $Y \ge s$.

$$
\Pr[\nepochs \ge z + \ell] 
  \le \Pr[Y \ge s] 
  \le \Pr[Y \ge r^{\ell}E[Y]] 
  \le \frac{1}{r^\ell}
$$
where we have used Markov's inequality.
 
Since $\nepochs$ takes only positive integral values,
\begin{eqnarray*}
E[\nepochs] &   = & \sum_{i>0} \Pr[\nepochs \ge i]  
                =  \sum_{i=1}^z \Pr[\nepochs \ge i]
                   + \sum_{\ell \ge 1} \Pr[\nepochs \ge z + \ell] \\
            & \le & z + \sum_{\ell \ge 1} \frac{1}{r^\ell} 
              \le z + \frac{1}{1-1/r}
              \le z + 2\\
\end{eqnarray*}
where we have assumed $r \ge 2$.
 
\end{proof}


Let $\mesg$ denote the total number of messages sent during the entire execution.  Let $\mesg_i$ denote the total number of messages sent in epoch $i$. Let $X_i$ denote the number of messages sent from the sites to the coordinator in epoch $i$.  $\mesg_i$ is the sum of two parts, (1)\ $k$ messages sent by the coordinator at the start of the epoch, and (2)\ two times the number of messages sent from the sites to the coordinator.

\begin{equation}
\label{eqn:mesgi}
\mesg_i = k + 2X_i
\end{equation}

\begin{equation}
\label{eqn:mesg}
\mesg = \sum_{j=0}^{\nepochs-1} \mesg_i = \nepochs k + 2\sum_{j=0}^{\nepochs-1} X_j
\end{equation}

For epoch $i$, consider the stochastic process $\cY = \{Y_j : j \ge 1\}$. For each $j$, choose a random number $w_j$ uniformly from $(0,1)$.
\[
Y_j = 
\begin{cases}
0 & \text{if  }  w_j \ge m_i \\
1 & \text{if  }  m_i/r < w_j < m_i \\
2 & \text{if  }  m_i/r \le w_j
\end{cases}
\]
Let $\tau$ denote the smallest time $t$ such that there are at least $s$ elements of $\{Y_1, Y_2,\ldots, Y_t\}$ that are equal to $2$. Let $Y = \sum_{j=1}^\tau Y_j$. 

\begin{lemma}
\label{lem:xiy}
\[
X_i < Y
\]
\end{lemma}

\begin{proof}
Consider the correspondence between the $j$th element received in epoch $i$ and $Y_j$. Each time a message is sent upon receiving the $j$th element, it must be true that $Y_j \ge 1$, since the random weight chosen $w_j$ must be less than $m_i$ for a message to be sent (note that the threshold could be stricter than $m_i$). Further, the number of elements in this epoch is less than or equal to $\tau$, since by the time $s$ elements are seen, each with a weight less than $m_i/r$, the epoch would have ended (it may have ended earlier).

\end{proof}

Consider the conditional random variables $Y_j(\alpha) = (Y_j | m_i=\alpha)$, $Y(\alpha) = (Y | m_i = \alpha)$, and $\tau(\alpha) = (\tau | m_i = \alpha)$.
\[
Y(\alpha) = \sum_{j=1}^{\tau(\alpha)} Y_j(\alpha)
\]

\begin{definition}
Let $\cZ = \{Z_n : n \ge 1\}$ be a stochastic process. A stopping time $\theta$ with respect to $\cZ$  is a random time such that for each $n \ge 0$, the event $\{\theta = n\}$ is completely determined by the total information known up to time $n$, i.e. $\{Z_1, Z_2,\ldots,Z_n\}$.
\end{definition}

\begin{theorem}[Wald's Equation]
\label{thm:wald}
If $\theta$ is a stopping time with respect to an i.i.d. sequence $\{Z_n : n \ge 1\}$ and if $E[\theta] < \infty$ and $E[|X|] < \infty$, then
\[
E\left[ \sum_{n=1}^\theta Z_n \right] = E[\theta] E[X]
\]
\end{theorem}

\begin{lemma}
\label{lem:y-alpha}
\[
E[Y(\alpha)] = (r+1)s
\]
\end{lemma}

\begin{proof}
\[
E[Y(\alpha)] = E\left[\sum_{j=1}^{\tau(\alpha)} Y_j(\alpha)\right]
\]
Note that the different $Y_j(\alpha)$ are independent and identically distributed since each $w_j$ is chosen independently from the same distribution. 
Further, $\tau(\alpha)$ is a stopping time for $\cY$, since for $n \ge 1$, the event $\tau(\alpha) = n$ can be determined by looking at the information till time $j$, i.e. $\{Y_1, Y_2,\ldots, Y_n\}$ and checking the number of $Y_j$s for $j < n$, that were equal to $2$.

Further, we note that $E[\tau(\alpha)]$ is finite, and $E[Y_j]$ is also finite. Using Wald's equation (Theorem~\ref{thm:wald}), we get
\[
E[Y(\alpha)] = E[\tau(\alpha)] E[Y_1(\alpha)] 
\]

Note that $\tau(\alpha)$ is the number of trials until $s$ successes, where the probability of a success is $\alpha/r$. Hence, $\tau(\alpha)$ is the sum of $s$ geometric random variables each with a parameter of $\alpha/r$, and $E[\tau(\alpha)] = sr/\alpha$.

From the definition of $Y_j$ and conditioning on $m_i = \alpha$, we have $Y_j = 0$ with probability $1-\alpha$, $1$ with probability $\alpha-\alpha/r$ and $2$ with probability $\alpha/r$. Hence:
\[
E[Y_1(\alpha)] = 0(1-\alpha) + 1(\alpha-\alpha/r) + 2\alpha/r = \alpha(1+1/r)
\]
Combining the above, the proof is complete.

\end{proof}

\begin{lemma}
\label{lem:xi}
\[
E[X_i] \le (r+1)s
\]
\end{lemma}

\begin{proof}
We have $E[Y] = E[E[Y | (m_i = \alpha)]] = E[(r+1)s] = (r+1)s$, where we used Lemma~\ref{lem:y-alpha}.
Using Lemma~\ref{lem:xiy}, the proof is complete.

\end{proof}

\begin{lemma}
\label{lem:msg-complexity}
\[
E[\mesg] \le (k + 2(r+1)rs) \left(\frac{\log(n/\samplesize)}{\log r} + 2\right)
\]
\end{lemma}

\begin{proof}
Using Lemma~\ref{lem:xi} and Equation \ref{eqn:mesgi}, we get the
expected number of messages in epoch $i$:
\[
E[\mesg_i] \le k + 2(r+1)s = k + 2s + 2rs
\]

Let $\cI\{\nepochs > i\}$ denote the indicator random variable that is $1$ when $\nepochs > i$ and $0$ otherwise. The total number of messages can be written as follows.
\[
\mesg = \sum_{i=1}^\infty \mesg_i \cI\{\nepochs > (i-1)\}
\]

Since $\mesg_i$ is independent of the event $\nepochs > (i-1)$, we have:
\begin{align*}
E[\mesg] & = \sum_{i=1}^\infty E[\mesg_i] E[\cI\{\nepochs > (i-1)\}] \\
               & \le (k+2s+2rs) \sum_{i=1}^\infty \Pr[\nepochs > (i-1)] \\
               & = (k+2s+2rs) E[\nepochs] \\
               & \le (k+2s+2rs) \left(\frac{\log(n/\samplesize)}{\log r} + 2\right)
\end{align*}
where we have used  Lemma~\ref{lem:num-epochs} for an upper bound on the expected number of epochs.
 
\end{proof}

\begin{theorem}
The expected message complexity $E[\mesg]$ of our algorithm is as
follows.
\begin{enumerate}
\item[I:]  If $s \ge \frac{k}{8}$, then
$E[\mesg] =  O\left(s  \log\left(\frac{n}{s}\right) \right)$

\item[II:] If $s < \frac{k}{8}$, then
$E[\mesg] = O\left(\frac{k \log\left(\frac{n}{s}\right)}{\log\left(\frac{k}{s}\right)}\right)$
\end{enumerate}
\end{theorem}

\begin{proof}
We note that the upper bounds on $E[\mesg]$ in Lemma
\ref{lem:msg-complexity} hold for any value of $r \ge 2$.

Case I: $s \ge \frac{k}{8}$. In this case, we set $r=2$. From Lemma
\ref{lem:msg-complexity},
\[
E[\mesg] \le (k+12s)  \left(\frac{\log(n/\samplesize)}{\log 2}\right) 
         \le   20s \log\left(\frac{n}{s}\right)
         =   O\left(s  \log\left(\frac{n}{s}\right) \right)
\]

Case II: $s < \frac{k}{8}$. We set $r = \frac{k}{s}$, and get:
\[
E[\mesg] = O\left(\frac{k \log\left(\frac{n}{s}\right)}{\log\left(\frac{k}{s}\right)}\right).
\]
 
\end{proof}

\section{Lower Bound}
\label{sec:lb}

\begin{theorem}\label{thm:lb}
For any constant $q, 0 < q < 1$, any correct protocol must send 
$\Omega \left (\frac{k \log (n/s)}{\log (1+(k/s))} \right )$ messages with probability 
at least $1-q$, where the probability is taken over the protocol's 
internal randomness. 
\end{theorem} 

\begin{proof}
Let $\beta = (1+(k/s))$. 
Define $e = \Theta\left (\frac{\log (n/s)}{\log (1+(k/s))} \right)$ epochs as follows: 
in the $i$-th epoch, $i \in \{0, 1, 2, \ldots, e-1\}$, 
there are $\beta^{i-1} k$ global stream updates, which
can be distributed among the $k$ servers in an arbitrary way.

We consider a distribution on orderings of the stream updates. Namely,
we think of a totally-ordered stream $1, 2, 3, \ldots, n$ of $n$ updates, and in the $i$-th
epoch, we randomly assign the $\beta^{i-1}k$ updates among the $k$ servers,
independently for each epoch. Let the randomness used for the assignment
in the $i$-th epoch be denoted $\sigma_i$. 

Consider the global stream of updates $1, 2, 3, \ldots, n$. Suppose
we maintain a sample set $\sample$ of $s$ items without replacement. We let 
$\sample_i$ denote a random variable indicating the value of $\sample$ after seeing
$i$ updates in the stream. We will use the following lemma
about reservoir sampling. 

\begin{lemma}
\label{lemma:rsampling}
For any constant $q > 0$, there is a constant $C' = C'(q) > 0$ for which
\begin{itemize}
\item $\sample$ changes at least $C's \log (n/s)$ times with probability
at least $1-q$, and
\item If $s < k/8$ and $k = \omega(1)$ and $e = \omega(1)$, then with probability at least $1-q/2$, 
over the choice of $\{\sample_i\}$, there are at least $(1-(q/8))e$ epochs for which
the number of times $\sample$ changes in the epoch is at least
$C's \log(1+(k/s))$.
\end{itemize}
\end{lemma}

\begin{proof}
Consider the stream $1, 2, 3, \ldots, n$ of updates. In the classical reservoir sampling
algorithm \cite{k81}, $\sample$ is initialized to $\{1, 2, 3, \ldots, s\}$. Then,
for each $i > s$, the $i$-th element is included in the current sample set $\sample_i$
with probability $s/i$, in which case a random item in $\sample_{i-1}$ is replaced with $i$.

For the first part of Lemma \ref{lemma:rsampling}, let $X_i$ be an indicator random variable
if $i$ causes $\sample$ to change. Let $X = \sum_{i=1}^n X_i$. 
Hence, ${\bf E}[X_i] = s/i$ for all $i$, and ${\bf E}[X] = H_n-H_s$, where $H_i = \ln i + O(1)$ is the $i$-th
Harmonic number. Then all of the $X_i$, $i > s$
are independent indicator random variables. It follows by a Chernoff
bound that 
$$\Pr[X < {\bf E}[X]/2] \leq \exp({\bf E}[X]/8) \leq \exp(-(\ln n/s)/8) \leq \left (\frac{s}{n} \right )^{1/8}.$$
For any $s = o(n)$, this is less than any constant $q$, and so the first part of Lemma
\ref{lemma:rsampling} follows since ${\bf E}[X]/2 = 1/2 \cdot \ln(n/s)$. 

For the second part of Lemma \ref{lemma:rsampling}, consider the $i$-th epoch, $i > 0$, 
which contains $\beta^{i-1}k$ consecutive
updates. Let $Y_i$ be the number of changes in this epoch. Then ${\bf E}[Y_i] = s \ln \beta + O(1)$. 
Since $Y_i$
can be written as a sum of independent indicator random variables, by a Chernoff bound,
$$\Pr[Y_i < {\bf E}[Y_i]/2] \leq \exp(-{\bf E}[Y_i]/8) 
\leq \exp(-(s \ln \beta + O(1))/8) 
\leq \frac{1}{\beta^{s/8}}.$$
Hence, the expected number of epochs $i$ for which $Y_i < {\bf E}[Y_i]/2$ is at most
$\sum_{i=1}^{e-1} \frac{1}{\beta^{s/8}}$, which is $o(e)$ since we're promised that $s < k/8$ and $k = \omega(1)$
and $e = \omega(1)$. 
By a Markov bound, with probability
at least $1-q/2$, at most $o(e/q) = o(e)$ epochs $i$ satisfy $Y_i \geq {\bf E}[Y_i]/2$. 
It follows that with probability at least $1-q/2$, 
there are at least $(1-q/8)e$ epochs $i$ for which the number $Y_i$ of changes in the epoch $i$ is at 
least ${\bf E}[Y_i]/2 \geq \frac{1}{2} s \ln \beta$, as desired. 

 
\end{proof}

{\bf Corner Cases:} When $s \geq k/8$, the statement of Theorem \ref{thm:lb} gives a lower bound
of $\Omega(s \log (n/s))$. In this case Theorem \ref{thm:lb} follows immediately
from the first part of Lemma \ref{lemma:rsampling} since these changes in $\sample$ must
be communicated to the central coordinator. Hence, in what follows we can assume $s < k/8$. 
Notice also that if $k = O(1)$,
then $\frac{k \log (n/s)}{\log (1 + (k/s))} = O(s \log (n/s))$, and so the theorem is 
independent of $k$, and follows simply by the first part of Lemma \ref{lemma:rsampling}.
Notice also that if $e = O(1)$, then the statement of Theorem \ref{thm:lb} amounts to proving an $\Omega(k)$ lower bound, which follows
trivially since every site must send at least one message. 

Thus, in what follows, we may apply
the second part of Lemma \ref{lemma:rsampling}.
\\\\
{\bf Main Case:} Let $C > 0$ be a sufficiently small constant, depending on $q$,
to be determined below.
Let $\Pi$ be a possibly randomized protocol, which with
probability at least $q$, sends at most $Cke$
messages. We show that $\Pi$ cannot be a correct protocol.

Let $\tau$ denote the random coin tosses of $\Pi$, i.e., the
concatenation of random strings of all $k$ sites together with that of
the central coordinator. 

Let $\mathcal{E}$ be the event that $\Pi$ sends less than 
$Cke$ messages. By assumption, $\Pr_{\tau}[\mathcal{E}] \geq q.$
Hence, it is also the case that
$$\Pr_{\tau, \{\sample_i\}, \{\sigma_i\}}[\mathcal{E}] \geq q.$$
For a sufficiently small constant $C' > 0$ that may depend on $q$, let
$\mathcal{F}$ be the event that there are at least $(1-(q/8))e$ epochs for
which the number of times $\sample$ changes in the epoch is at least $C' s \log(1+(k/s))$.
By the second part of Lemma \ref{lemma:rsampling},
$$\Pr_{\tau, \{\sample_i\}, \{\sigma_i\}}[\mathcal{F}] \geq 1-q/2.$$
It follows that there is a fixing of $\tau = \tau'$ as well as a fixing of 
$\sample_0, \sample_1, \ldots, \sample_e$ to $P_0', P_1', \ldots, P_e'$
for which $\mathcal{F}$ occurs and
$$\Pr_{\{\sigma_i\}}[\mathcal{E} \mid \tau = \tau', \ 
(\sample_0, \sample_1, \ldots, \sample_e) = (P_0', P_1', \ldots, P_e')] \geq q-q/2 = q/2.$$
Notice that the three (sets of) random variables $\tau, \{P_i\},$
and $\{\sigma_i\}$ are independent, and so in particular, 
$\{\sigma_i\}$ is still uniformly random given this conditioning.

By a Markov argument, if event $\mathcal{E}$ occurs, then there are at least 
$(1-(q/8))e$ epochs for which at most 
$(8/q) \cdot C \cdot k$ messages are sent. If events $\mathcal{E}$ and $\mathcal{F}$ 
both occur, then by 
a union bound, there are at least $(1-(q/4))e$ epochs for which at most 
$(8/q) \cdot C \cdot k$ messages
are sent and $S$ changes in the epoch at least $C's \log (1+(k/s))$ times. 
Call such an epoch {\it balanced}. 

Let $i^*$ be the epoch which is most likely
to be balanced, over the random choices of $\{\sigma_i\}$, conditioned on
$\tau = \tau'$ and $(\sample_0, \sample_1, \ldots, \sample_e) = (P_0', P_1', \ldots, P_e')$. Since 
at least $(1-(q/4))e$ epochs are balanced if $\mathcal{E}$ and $\mathcal{F}$ occur,
and conditioned on $(\sample_0, \sample_1, \ldots, \sample_e) = (P_0', P_1', \ldots, P_e')$ event
$\mathcal{F}$ does occur, and $\mathcal{E}$ occurs with probability at least $q/2$
given this conditioning, it follows that
$$\Pr_{\{\sigma_i\}}[i^* \textrm{ is balanced } \mid \ \tau = \tau', \ 
(\sample_0, \sample_1, \ldots, \sample_e) = (P_0', P_1', \ldots, P_e')] \geq q/2-q/4 = q/4.$$
The property of $i^*$ being balanced is independent of $\sigma_j$ for $j \neq i^*$,
so we also have
$$\Pr_{\sigma_{i^*}}[i^* \textrm{ is balanced } \mid \ \tau = \tau', \ 
(\sample_0, \sample_1, \ldots, \sample_e) = (P_0', P_1', \ldots, P_e')] \geq q/4.$$
If $C's \log(1+(k/s)) \geq 1$, then $\sample$ changes at least once in epoch $i^*$. 
Suppose, for the moment, that this is the case. Suppose the first
update in the global stream at which $\sample$ changes is the $j^*$-th update. In
order for $i^*$ to be balanced for at least a $q/4$ fraction of the $\sigma_{i^*}$, 
there must be at least $qk/4$ different servers
which receive $j^*$, for which $\Pi$
sends a message. In particular, since $\Pi$ is deterministic conditioned on $\tau$,
at least $qk/4$ messages must be sent in the $i^*$-th epoch. But $i^*$ was chosen
so that at most $(8/q) \cdot C \cdot k$ messages are sent, which is a contradiction
for $C < q^2/32$. 

It follows that we reach a contradiction unless $C's \log (1+(k/s)) < 1$. Notice, though,
that since $C'$ is a constant, if $C's \log (1+(k/s)) < 1$, 
then this implies that $k = O(1)$. However, if $k = O(1)$,
then $\frac{k \log (n/s)}{\log (1 + (k/s))} = O(s \log (n/s))$, and so the theorem is 
independent of $k$, and follows simply by the first part of Lemma \ref{lemma:rsampling}.

Otherwise, we have reached a contradiction, and so it follows that $Cke$ messages must 
be sent with probability at least $1-q$. Since
$Cke = \Omega \left (\frac{k \log (n/s)}{\log (1+ (k/s))} \right )$, this completes the proof. 
\end{proof}

\section{Sampling With Replacement}
\label{sec:withrep}

We now present an algorithm to maintain a random sample of size
$s$ with replacement from $\stream$. The basic idea is to
run in parallel $s$ copies of the single item sampling algorithm from
Section \ref{sec:algo}. Done naively, this will lead to a message complexity of 
$O(s k \frac{\log n}{\log k})$. We obtain an improved algorithm based on the
following ideas.

We view the distributed streams as $s$ logical streams, $\stream^i, i
= 1\ldots s$.  Each $\stream^i$ is identical to $\stream$, but the
algorithm assigns independent weights to the different copies of the
same element in the different logical streams. Let $w^i(e)$ denote the
weight assigned to element $e$ in $\stream^i$. $w^i(e)$ is a random
number between $0$ and $1$. For each $i=1 \ldots s$, the coordinator
maintains the minimum weight, say $w^i$, among all elements in
$\stream^i$, and the corresponding element.

Let $\beta = \max_{i=1}^s w^i$; $\beta$ is maintained by the
coordinator. Each site $j$ maintains $\beta_j$, a local view of
$\beta$, which is always greater than or equal to $\beta$. Whenever a
logical stream element at site $j$ 
has weight less than $\beta_j$, the site sends
it to the coordinator, receives in response the current value of
$\beta$, and updates $\beta_j$. When a random sample is requested at
the coordinator, it returns the set of all minimum weight elements in
all $s$ logical streams. It can be easily seen that this algorithm
is correct, and at all times, returns a random sample of size $s$
selected with replacement.
The main optimization relative to the naive approach described above
is that when a site sends a message to the coordinator, it receives
$\beta$, which provides partial information about all $w^i$s. This
provides a substantial improvement in the message complexity and leads
to the following bounds. 

\begin{theorem}
\label{thm:withrep}
The above algorithm continuously maintains a
sample of size $s$ with replacement from $\stream$, and its expected
message complexity is $O(s \log s \log n)$ in case $k \le 2 s \log s$, and 
$O\left(k \frac{\log n}{\log(\frac{k}{s\log s})} \right)$ in case 
$k > 2 s \log s$.
\end{theorem}

\begin{proof}
We provide a sketch of the proof here.  The analysis of the message
complexity is similar to the case of sampling without replacement. We
sketch the analysis here, and omit the details. The execution is
divided into epochs, where in epoch $i$, the value of $\beta$ at the
coordinator decreases by at least a factor of $r$ (a parameter to be
determined later). Let $\nepochs$ denote the number of epochs. It can
be seen that $E[\nepochs] = O(\frac{\log n}{\log r})$. In epoch $i$,
let $X_i$ denote the number of messages sent from the sites to the
coordinator in the epoch, $m_i$ denote the value of $\beta$ at the
beginning of the epoch, and $n_i$ denote the number of elements in
$\stream$ that arrived in the epoch.

The $n_i$ elements in epoch $i$ give rise to $s n_i$ logical elements,
and each logical element has a probability of no more than $m_i$ of
resulting in a message to the coordinator. Similar to the proof of
Lemma \ref{lem:xi}, we can show using conditional expectations that
$E[X_i] \le r s \log s$ (the $\log s$ factor comes in due to the
fact that $E[n_i | m_i = \alpha] \le \frac{r \log s}{\alpha}$. Thus
the expected total number of messages in epoch $i$ is bounded by 
$(k + 2sr \log s)$, and in the entire execution is 
$O((k + 2sr \log s)\frac{\log n}{\log r})$. By choosing $r = 2$ for the
case $k \le (2s \log s)$, and $r=k/(s \log s)$ for the case 
$k > (2s \log s)$, we get the desired result.
 
\end{proof}

\remove{
Note that the message complexity of the algorithm for sampling
with replacement from \cite{cmyz10} is $O(k + s \log s) \log n$. In
comparison with this, our algorithm does better by a factor of 
$O(\log(\frac{k}{s\log s}))$ in case $k > 2s \log s$. 
}

\remove{

\section{Discussion}
Our model of assuming that no more than one element arrives at each
site in each round is somewhat artificial, but is needed because the
problem asks for the central coordinator to always maintain a random
sample from among all the elements seen so far. Thus, if the arriving
element is selected to be a part of the random sample, it must be
immediately communicated to the coordinator before any further
elements arrive.  In practice, communication is not synchronous, and
we will have to relax the requirement of the coordinator always having
a random sample of all the stream elements observed until the current
instant. Then it is possible for a site to observe multiple elements
(more than one of which may be sampled) before communicating with the
coordinator. Even in this case, our algorithms work with minor
modifications.
}

\nocite{Vitt85,muthu-book,XTB08,BO10,cmyz10}

\bibliographystyle{abbrv}
\bibliography{distsamp}

\begin{thebibliography}{10}

\bibitem{A06}
C.~C. Aggarwal.
\newblock On biased reservoir sampling in the presence of stream evolution.
\newblock In {\em VLDB}, pages 607--618, 2006.

\bibitem{abc09}
C.~Arackaparambil, J.~Brody, and A.~Chakrabarti.
\newblock Functional monitoring without monotonicity.
\newblock In {\em ICALP (1)}, pages 95--106, 2009.

\bibitem{BDM02}
B.~Babcock, M.~Datar, and R.~Motwani.
\newblock Sampling from a moving window over streaming data.
\newblock In {\em SODA}, pages 633--634, 2002.

\bibitem{bo03}
B.~Babcock and C.~Olston.
\newblock Distributed top-k monitoring.
\newblock In {\em SIGMOD Conference}, pages 28--39, 2003.

\bibitem{BO10}
V.~Braverman and R.~Ostrovsky.
\newblock Effective computations on sliding windows.
\newblock {\em SIAM Journal on Computing.}, 39(6):2113--2131, 2010.

\bibitem{BOZ09}
V.~Braverman, R.~Ostrovsky, and C.~Zaniolo.
\newblock Optimal sampling from sliding windows.
\newblock In {\em PODS}, pages 147--156, 2009.

\bibitem{cg05}
G.~Cormode and M.~N. Garofalakis.
\newblock Sketching streams through the net: Distributed approximate query
  tracking.
\newblock In {\em VLDB}, pages 13--24, 2005.

\bibitem{cmy08}
G.~Cormode, S.~Muthukrishnan, and K.~Yi.
\newblock Algorithms for distributed functional monitoring.
\newblock In {\em SODA}, pages 1076--1085, 2008.

\bibitem{cmyz10}
G.~Cormode, S.~Muthukrishnan, K.~Yi, and Q.~Zhang.
\newblock Optimal sampling from distributed streams.
\newblock In {\em PODS}, pages 77--86, 2010.

\bibitem{DN06}
M.~Dash and W.~Ng.
\newblock Efficient reservoir sampling for transactional data streams.
\newblock In {\em Sixth IEEE International Conference on Data Mining
  (Workshops)}, pages 662 --666, 2006.

\bibitem{GT01}
P.~B. Gibbons and S.~Tirthapura.
\newblock Estimating simple functions on the union of data streams.
\newblock In {\em SPAA}, pages 281--291, 2001.

\bibitem{GT02}
P.~B. Gibbons and S.~Tirthapura.
\newblock Distributed streams algorithms for sliding windows.
\newblock In {\em SPAA}, pages 63--72, 2002.

\bibitem{HNGHJJT07}
L.~Huang, X.~Nguyen, M.~N. Garofalakis, J.~M. Hellerstein, M.~I. Jordan, A.~D.
  Joseph, and N.~Taft.
\newblock Communication-efficient online detection of network-wide anomalies.
\newblock In {\em INFOCOM}, pages 134--142, 2007.

\bibitem{kcr06}
R.~Keralapura, G.~Cormode, and J.~Ramamirtham.
\newblock Communication-efficient distributed monitoring of thresholded counts.
\newblock In {\em SIGMOD Conference}, pages 289--300, 2006.

\bibitem{k81}
D.~E. Knuth.
\newblock {\em The Art of Computer Programming, Volume II: Seminumerical
  Algorithms, 2nd Edition}.
\newblock Addison-Wesley, 1981.

\bibitem{MV07}
V.~Malbasa and S.~Vucetic.
\newblock A reservoir sampling algorithm with adaptive estimation of
  conditional expectation.
\newblock In {\em IJCNN 2007, International Joint Conference on Neural
  Networks}, pages 2200 --2204, 2007.

\bibitem{msdo05}
A.~Manjhi, V.~Shkapenyuk, K.~Dhamdhere, and C.~Olston.
\newblock Finding (recently) frequent items in distributed data streams.
\newblock In {\em ICDE}, pages 767--778, 2005.

\bibitem{muthu-book}
S.~Muthukrishnan.
\newblock {\em Data Streams: Algorithms and Applications}.
\newblock Foundations and Trends in Theoretical Computer Science. Now
  Publishers, August 2005.

\bibitem{Vitt85}
J.~S. Vitter.
\newblock Random sampling with a reservoir.
\newblock {\em ACM Transactions on Mathematical Software}, 11(1):37--57, 1985.

\bibitem{XTB08}
B.~Xu, S.~Tirthapura, and C.~Busch.
\newblock Sketching asynchronous data streams over sliding windows.
\newblock {\em Distributed Computing}, 20(5):359--374, 2008.

\bibitem{yz09}
K.~Yi and Q.~Zhang.
\newblock Optimal tracking of distributed heavy hitters and quantiles.
\newblock In {\em PODS}, pages 167--174, 2009.

\end{thebibliography}

\end{document}